\documentclass[sn-mathphys]{sn-jnl}
\usepackage{amssymb,bbm} 
\usepackage{amsmath} 
\usepackage{mathtools}

\jyear{2021}%

\theoremstyle{thmstyleone}%
\newtheorem{theorem}{Theorem}
\newtheorem{corollary}[theorem]{Corollary}
\newtheorem{proposition}[theorem]{Proposition}%

\theoremstyle{thmstyletwo}%
\newtheorem{remark}{Remark}%

\theoremstyle{thmstylethree}%
\newtheorem{definition}{Definition}%
\newcommand{\rd}{{\rm d}}
\newcommand{\rP}{{\rm P}}
\newcommand{\rPi}{{\rm P}_{\rm inv}}
\newcommand{\Ji}{J_{\rm inv}}
\newcommand{\rhoi}{\rho_{\rm inv}}
\newcommand{\dsig}{\dot{\Sigma}_{\rm s}}
\newcommand{\avg}[1]{\mathbb{E}^{\rm inv}_{\Delta t}\left( #1 \right)}
\newcommand{\avgt}[1]{\mathbb{E}^{\rm inv}_{t}\left( #1 \right)}

\newcommand{\blue}[1]{\textcolor{black}{#1}}

\begin{document}

\title[Upper bounds on dissipation in irreversible diffusions]{Vorticity and level-set variations of invariant current
bound steady-state dissipation}

\author[1,2]{\fnm{Hao} \sur{De}}\email{hao.de@queens.ox.ac.uk}

\author*[2]{\fnm{Alja\v{z}} \sur{Godec}}\email{agodec@mpinat.mpg.de}

\affil[1]{\orgdiv{Mathematical Institute}, \orgname{University of Oxford},
  \orgaddress{\street{Woodstock Road}, \city{Oxford}, \postcode{OX2 6GG}, \country{United Kindgom}}}

\affil*[2]{\orgdiv{Mathematical bioPhysics Group}, \orgname{Max Planck
    Institute for Multidisciplinary Sciences}, \orgaddress{\street{Am
      Fassberg 11}, \city{G\"ottingen}, \postcode{37077}, \country{Germany}}}


\abstract{A non-vanishing entropy production rate is a hallmark of non-equilibrium
  stationary states and is therefore at the heart of non-equilibrium
  thermodynamics. It is a manifestation of a steady circulation
  $J_{\rm inv}$ along the
  level sets of the invariant density $\rho_{\rm inv}$, and is thus generically used to
  quantify how far a steady system is driven out of
  equilibrium. While it is well known that there exists a continuum of
  distinct steady states with 
  the same invariant measure, the question how the geometry and topology
  of the invariant current \emph{a priori} affect dissipation
  remained elusive. For confined irreversible diffusions we identify two minimal descriptors, the
  $\rho_{\rm inv}$-weighted vorticity and the variation of $J_{\rm
    inv}$ along level sets of $\rho_{\rm inv}$, and prove that
  these jointly bound
  from above the steady-state entropy production rate. 
  In regions where $\rhoi$ is close to Gaussian the bound is
dominated solely by the vorticity of the drift field and 
    in the low-noise (Freidlin–Wentzel) limit 
    by any non-potential contribution to the drift, 
rendering $J_{\rm inv}$ virtually 
  constant along the level sets of $\rho_{\rm inv}$.}




\maketitle

\section{Introduction}\label{sec1}
A defining feature of non-equilibrium as opposed to equilibrium steady states is a
positive entropy production
\cite{Jiang2004,Maes2008PA,Maes2008EEL,Seifert2010EEL,Barato2015PRL,Gingrich2016PRL,Dieball2022PRL,Dieball2022PRR}. The
thermodynamics of far-from-equilibrium systems is essential for
elucidating the physical
principles that sustain driven, in particular living matter
\cite{Active,Active_1,Active_2,Active_3,Active_4}. Erwin
Schrödinger observed already in 1944 \cite{Erwin}
that living organisms must continuously increase the entropy of its
surroundings in order to avoid decaying to thermodynamic
equilibrium. Far-from-equilibrium thermodynamics was put on formal
grounds by the ``Brussels school'' \cite{Prigogine} and was more
recently extended to Stochastic Thermodynamics \cite{Seifert2012RPP} focussing on a
sample-path perspective. 

Stochastic Thermodynamics provides a quite general framework that, in
principle, applies arbitrarily far from equilibrium
\cite{Seifert2012RPP}. The main constraining assumption of
Stochastic Thermodynamics is that of \emph{local equilibrium}, stating
that all unobserved degrees of freedom evolve sufficiently fast with
respect to the observed ones, such that they may be faithfully considered to be,
at every instance, at thermodynamic equilibrium (i.e.\ ``Gibbsian'') at the temperature of
the surroundings (also called ``the heat bath'')
\cite{Jiang2004,Spohn_LDB,Spohn_book,Eckmann,Lefevre,Landim_LDB,Jourdain,Maes_LDB}. Within
this local-equilibrium paradigm the total entropy flow along an
individual \emph{stochastic} sample path $\omega$ compares the probability measure of
a path, $\rP_t(\omega)$, with the probability measure of the respective
time-reversed path,  $\rP_{\theta t}(\omega)$, both defined on the
$\sigma$-algebra $\mathcal{F}^t_\tau$ generated by the process $(x_\tau)_{0\le \tau\le t}$, and is defined as
\begin{equation*}
\Sigma_t(\omega)\equiv\ln \frac{\rd\rP_t}{\rd\rP_{\theta t}}(\omega)
  \label{def}
\end{equation*}
where $\theta (x_\tau)_{0\le \tau\le t}\equiv \epsilon (x_{t-\tau})_{0\le \tau\le t}$ denotes the
time-reversal operation and, properly defined for the given setting;
in particular, $\epsilon$ reverses the sign of all degrees of freedom
with odd time parity 
(see
e.g.\ \cite{Jiang2004,Spohn_LDB,Spohn_book,Eckmann,Lefevre,Landim_LDB,Jourdain,Maes_LDB}
for various settings). In typical physically relevant settings \cite{Jourdain,Jiang2004}, and in
particular in the case of overdamped diffusions considered herein (see  \cite{Qian_CMP})
$\rP(\omega)$ may be assumed to possess a density $p(x_\tau(\omega))$
and $\Sigma(\omega)$ is defined via the log-ratio of the respective
probability densities $p(x_\tau(\omega))$ and $p(\theta
x_\tau(\omega))$.

In this work we focus exclusively on strongly ergodic dynamics with
invariant measure $\rPi$ that also possesses a density $\rhoi$ and, as
a result of broken detailed balance,  carries a stationary invariant
current $\Ji$. In this
setting the thermodynamic observable of central interest is the
\emph{steady state entropy production rate} (or steady-state
dissipation) defined as
\begin{equation}
\dsig\equiv \lim_{t\to\infty}\frac{1}{t}\ln
\frac{\rd\rP_t}{\rd\rP_{\theta
    t}}(\omega)=\lim_{\Delta t\to 0^+}\frac{1}{\Delta
  t}\avg{\ln\frac{p([x_\tau]_{0\le \tau\le \Delta t})}{p([\theta x_\tau]_{0\le \tau\le \Delta t})}},
\label{epr}  
\end{equation}
which characterizes the degree of irreversibility and hence
non-equilibrium and where the second equality follows from $x_\tau$
being ergodic. $\avgt{\cdot}$ denotes the expectation over $p([
  x_\tau]_{0\le \tau\le t})$ in the steady-state ensemble,
i.e.\ over paths $(x_\tau)_{0\le \tau\le t}$ propagating from the
initial density $p(x_{0})=\rhoi(x_0)$ up to time $t$. The limits in
Eq.~\eqref{epr} exist \cite{Qian_CMP}.

We consider overdamped diffusions evolving according to the 
stochastic differential equation \cite{Ikeda1981}
\begin{equation}
dx_\tau=F(x_\tau)d\tau+\sqrt{2D}dW_\tau
\label{SDE}  
\end{equation}
with $x_\tau\in\mathbb{R}^d$, a smooth confining drift field
$F:\mathbb{R}^d\to\mathbb{R}^d$, the $d\times d$, symmetric
positive
definite matrix $D$, and $W_\tau$ denoting
the $d$-dimensional Wiener process. $F$ is supposed to have nice
growth properties and to locally satisfy Lipschitz conditions guaranteeing the
existence and uniqueness of a solution of Eq.~\eqref{SDE} \cite{Ikeda1981}.
The square root of $\sqrt{D}$ is defined as usual
as $S^T\sqrt{\lambda_i}S$ via orthogonal diagonalization  $SDS^T={\rm
  diag}(\lambda_i)$ with $S^T=S^{-1}$.

From a physical perspective
Eq.~\eqref{SDE} embodies stochastic evolutions in low Reynolds media
such as macromolecules \cite{Doi1988}, molecular machines \cite{Motors}, or
suspensions of colloidal particles \cite{Dhont}.

The Ito equation \eqref{SDE} translates to the Fokker-Planck (or forward
Kolmogorov) equation for the 
probability density $P:\mathbb{R}^d\times[0,\infty)\to\mathbb{R}^+$
\begin{equation}
\partial_t P(x,t) =\nabla \cdot [D\nabla -F(x)] P(x,t)\equiv -\nabla \cdot J(x,t)
\label{FPE}
\end{equation}
where $\nabla$ is the $d$-dimensional gradient operator and we defined the current density $J(x,t)\equiv(F(x)-D\nabla)P(x,t)$, $J:\mathbb{R}^d\times[0,\infty)\to\mathbb{R}^d$. By
our assumptions we have $\lim_{t\to\infty}P(x,t)=\rhoi(x)$ and
$\lim_{t\to\infty}J(x,t)=\Ji(x)$ and obviously $\nabla\cdot\Ji(x)=0$,
i.e.\ the invariant current is incompressible.

As a result of the incompressibility of the stationary current,
the drift $F$ of the ergodic diffusion in Eq.~\eqref{SDE} may be
decomposed into reversible $F_{\rm rev}$ and irreversible $F_{\rm irr}$
\cite{Dieball2022PRR,Duong_2023} (alternatively called
``conservative'' and ``dissipative'' \cite{Graham,Eyink}) components,
$F(x)=F_{\rm rev}(x)+F_{\rm irr}(x)$ defined as
\begin{eqnarray}\label{Fdecomp}
  F_{\rm rev}(x)&\equiv&D\nabla\ln\rhoi(x)\equiv -D\nabla\Phi(x)\nonumber\\
  F_{\rm irr}(x)&\equiv&\frac{\Ji(x)}{\rhoi(x)}=
  A(x)\nabla\ln\rhoi(x)+\nabla\cdot A(x),
  \label{decomp}
\end{eqnarray}
where $\Phi(x)$ is the generalized ``nonequilibrium'' potential, $A(x)$ is a $d\times d$ anti-symmetric matrix ($A(x)^T=-A(x)$)
and we define the divergence of a matrix as $(\nabla\cdot
A(x))_k\equiv\sum_{l=1}^d\partial_{x_l}A_{kl}(x)$. One can henceforth
show (see e.g.~\cite{Seifert2012RPP,Qian_CMP}) that
\begin{equation}
\dot{\Sigma}_s=\int dx F_{\rm irr}(x)\cdot D^{-1}F_{\rm irr}(x)\rhoi(x)\ge 0
\label{dissipation}  
\end{equation}
The intuition about
$A(x)$ in Eq.~\eqref{decomp} is that it describes variations of the
$\Ji(x)$ along the level sets of $\rhoi(x)$. That is, if $A(x)$ is a
constant matrix (i.e.\ $\nabla\cdot A=0$) then $\lvert \Ji(x)\rvert$ is constant on
level sets of $\rhoi(x)$ and $F_{\rm rev}(x)\perp F_{\rm
  irr}(x),\,\forall x$. However, when $A(x)$ is not constant (which is
typically the case), $F_{\rm
  irr}(x)$ has a component in the direction of $F_{\rm
  rev}(x)$, contained in $\nabla\cdot A(x)$, and it is generally impossible to
determine $A(x)$ from $F(x)$, and thus steady-state dissipation $\dot{\Sigma}_s$ via Eq.~\eqref{dissipation}, without first determining $\rhoi(x)$, and
there is no general way to do so. Note that a Helmholtz-Hodge type may
deliver mutually orthogonal fields, but these will not be $F_{\rm
  rev}(x)$ and $F_{\rm irr}(x)$. Therefore, seemingly not as much can be inferred
about the nonequilibrium thermodynamics of the process \eqref{SDE}
directly from $F(x)$ alone. Moreover, even in cases when both
$\rhoi(x)$ and $F_{\rm  irr}(x)$ are known, an understanding of how
the properties of $\rhoi(x)$ and $F_{\rm  irr}(x)$, and in particular the
topology of $\Ji(x)$, affect $\dot{\Sigma}_s$ remains elusive.   

Notwithstanding its physical importance, $\dot{\Sigma}_s$ is very
difficult to quantify in general, as it requires detailed knowledge about \emph{all}
dissipative degrees of freedom, i.e.\ the knowledge of $\rhoi(x)$ and $\Ji(x)$. Important insight can be obtained in
the form of experimentally accessible lower bounds via the so-called
Thermodynamic Uncertainty Relation
\cite{Barato2015PRL,Pietzonka2016PRE,Dechant_2018,Horowitz2019NP,Koyuk2020PRL,Falasco2020NJP,Dechant2021PRX,Dieball_TUR},
which bounds $\dot{\Sigma}_s$ from below by means of any
time-accumulated current observable inferred from individual
stochastic trajectories. This approach indeed provides a very powerful means of
inferring $\dot{\Sigma}_s$, but does not give insight into what
properties of $F(x)$ relate to $\dot{\Sigma}_s$ and in what
manner. For this, one needs a more direct connection between the two
that does not require the knowledge of $\rhoi(x)$ as an input.    

An important rigorous result in this direction was obtained in terms of the
\emph{circulation decomposition}, the sum of an
``observed'' and ``hidden circulation'' around some closed curve and
note that both contributions are positive
\cite{Qian_CMP} (see also a more recent work on a related topic \cite{Ramin}). This approach, notably, does not require to know
$\rhoi(x)$ nor $\Ji(x)$. However, it depends explicitly on the choice
of the closed curve and therefore can in practice serve as a
lower bound on $\dot{\Sigma}_s$, as the possibility of missing the
``hidden'' circulation for a selected closed curve cannot be
dismissed if one does not know the topology of $\Ji(x)$.      

There is thus a need for results that directly relate the properties
of $F(x)$ to $\dot{\Sigma}_s$ in a general setting. Here we make
progress in this direction by proving upper bounds on $\dot{\Sigma}_s$
in terms of the vorticity of $F(x)$ and variation of $\Ji(x)$ on level
sets of $\rhoi(x)$ by using Sobolev inequalities.

In the (low-noise) Freidlin-Wentzel limit, i.e.\ $\|D\|\to 0$
  we show that the bound on $\dot{\Sigma}_s$ becomes dominated by \blue{
$\nabla\Phi(x)\times F_{\rm irr}(x)$ alone, which 
excludes any contributions parallel to $\nabla\Phi(x)$, and in fact if $F_{\rm rev}(x)$ is parallel to $\nabla\Phi(x)$ then $F_{\rm irr}(x)$ may be replaced by $F(x)$.} Moreover, in
regions where $\rhoi(x)$ is essentially a Gaussian (or can be well
approximated by one) the upper bound on $\dot{\Sigma}_s$ is as well
dominated solely by vorticity. Notably, as forward and backward paths between
pairs of points are nominally expected to be different in vortex flow,
the fact that vorticity bounds time-irreversibility from above seems
quite intuitive. Our results provide a
direct and intuitive link between the topology of $\Ji(x)$ and the steady-state
dissipation $\dot{\Sigma}_s$. By bounding  $\dot{\Sigma}_s$ from above,
they are complementary to the lower bounds established by the
Thermodynamic Uncertainty relation. The results in the low-noise limit, moreover, provide
further support for the so-called ``local detailed balance'' structure
of transition rates assumed in constructing approximate Markov-jump
representations of dynamic on long time-scales.\\
\indent The paper is structured as follows. In Sec.~\ref{sec2} we
present a concise summary of our main results and discuss their
physical implications. In Sec.~\ref{sec3} we then give the precise
definitions and state the essential preparatory theorems. In
Sec.~\ref{sec4} we present the detailed proofs of the results.

\section{Summary of the main results}\label{sec2}
We throughout assume that $x_\tau\in\Omega$ and that the domain $\Omega\subset \mathbb{R}^2\text{ or
}\mathbb{R}^3$ is bounded,
Lipschitz-continuous and simply-connected. We seek insight into which
properties of $\Ji(x)$ and in turn $F_{\rm irr}(x)$ of the overdamped
diffusion in Eq.~\eqref{SDE} determine the steady-state
dissipation $\dot{\Sigma}$ defined in Eq.~\eqref{dissipation}.

We will show that these turn out to be the  $\rhoi$-weighted \emph{vorticity}
\begin{equation*}
V^q\equiv\int_\Omega \lvert \nabla\times F_{\rm irr}(x)\lvert^{2q}
\rhoi(x)dx,
\end{equation*}
where $q\ge 1$ and with $V^1\equiv V$, the $\rhoi$-weighted level-set variations of $\Ji(x)$
\begin{equation*}
\delta^q_{\rm LS}\equiv\int_{\Omega} \lvert \nabla\cdot F_{\rm irr} \rvert^{2q} \rhoi(x) dx=\int_\Omega \lvert \nabla \Phi(x) \cdot  F_{\rm irr}(x)\lvert^{2q}
\rhoi(x)dx,
\end{equation*}
with $\delta^1_{\rm LS}\equiv\delta_{\rm LS} $, and the  $\rhoi$-weighted deviations of $\Ji(x)$ from alignment with $\nabla\Phi(x)$
\begin{equation*}
\delta^q_{\nparallel \nabla \Phi}\equiv\int_\Omega \lvert \nabla \Phi(x) \times  F_{\rm irr}(x)\lvert^{2q}
\rhoi(x)dx,
\end{equation*}
with $\delta^1_{\nparallel \nabla \Phi}\equiv \delta_{\nparallel \nabla \Phi}$.

The vorticity $V^q$ is defined in a standard fashion and the $\rhoi$ weight
is intended to ``grade'' the local vorticity features according to
their physical relevance in the steady-state ensemble. When
  $D={\rm const}\cdot \mathbbm{1}$ we have $\nabla \times D\nabla\Phi(x)=0$ and
  we may replace $\nabla\times F_{\rm irr}(x)\to \nabla\times F(x)$ everywhere. 

To provide intuition about why $\delta^q_{\rm LS}$ measures level-set variations of
$\Ji(x)$ we note that
if $\Ji(x)$ flows
perfectly along the level sets of $\rhoi(x)$ and the level sets are closed curves, then the conservation of mass requires $\lvert \Ji(x) \rvert$ must be constant along each level set. 
As $\rhoi(x)$ is by definition constant on level sets, $\Ji(x)$ being mass-preserving is equivalent to $F_{\rm irr}(x)$ being volume-preserving i.e.\ divergence free. 
Hence, any level-set variation of $\lvert \Ji(x)\rvert$ leads to
non-zero $\nabla\cdot F_{\rm irr}(x)$ and non-zero $\nabla\Phi(x)\cdot F_{\rm irr}(x)$.

Just as $\delta^q_{\rm LS}$ essentially measures deviations of $\Ji(x)$ from being perfectly
orthogonal to $\nabla\Phi(x)$ on level sets of $\rhoi$, 
the quantity $\delta^q_{\nparallel \nabla \Phi}$ measures how much $\Ji(x)$ deviates from being parallel to the potential gradient $\nabla\Phi(x)$.

Under these assumptions (that will be detailed further below) we prove the
following upper bounds.

\begin{enumerate}

\item For a general diffusion \eqref{SDE} in nonequilibrium steady
  state we find (see Theorem~\ref{theorem1})
\begin{align*}
\dot{\Sigma}_s&\le C(\Omega) \left [V + \frac{1}{4}(\delta_{\rm
    LS}+\delta_{\nparallel \nabla \Phi})\right ]\,,
\end{align*}  
where $C(\Omega)$ is a constant depending only of $\Omega$.\\

\item In the Freidlin-Wentzel limit when $\|D\|=\mathcal{O}(\epsilon)\to 0$ in
  Eq.~\eqref{SDE}, we prove (see
  Corollary~\ref{Freidlin-Wentzel}) that to leading order in  $\epsilon$ 
\begin{align*}
 & \dot{\Sigma}^\epsilon_s \leq C(\Omega) \int_{\Omega}
  \lvert\nabla\Phi^\epsilon\times F^\epsilon_{\rm irr}(x)\rvert^2\rhoi^\epsilon \ dx,   
\end{align*}
where $C(\Omega)$ is a constant depending only on the domain
$\Omega$. Moreover, \blue{if the reversible drift is parallel to the potential gradient}, we may
  replace $F^\epsilon_{\rm irr}(x)$ by \emph{full} drift
field $F(x)$, suggesting that as $\|D\|\to 0$ (i.e.\ in the
low-noise limit) dissipation is bounded from
above directly by any non-potential contribution to $F$ in regions
where $\rhoi^\epsilon$ has substantial mass.\\

\item When $\underset{\Omega}{\sup}\,\rhoi<\infty$ and
  $\rhoi^{1/(1-q)}\in L^1(\Omega)$ for $1<q<\infty$ we show (see
  Theorem~\ref{theorem2})
  \begin{align*}
    \dot{\Sigma}_s\leq C(\Omega,\rhoi) \left[(V^q)^{1/q}+
      (\delta^q_{\rm LS})^{1/q}
      \right]
  \end{align*}
where 
\begin{align*}
    C(\Omega,\rhoi)&=C(\Omega)\cdot \sup_\Omega \rhoi\cdot  \left(\int_\Omega \rhoi^{1/(1-q)}dx \right)^{(q-1)/q}.
\end{align*}
is a constant depending only on $\Omega$ and $\rhoi$.\\

\item Moreover, if in addition $\underset{\Omega}{\inf}\,\rhoi>0$ we
  have that (see Corollary~\ref{bounds})
\begin{align*}
    \dot{\Sigma}_s&\leq C(\Omega)\, \sup_\Omega \rhoi\, \left(\inf_\Omega \rhoi \right)^{-1}  \left[V+\delta_{\rm LS}\right],
\end{align*}
where the constant $C(\Omega)$ depends only on $\Omega$.\\

\item If on a subdomain $\omega\subset\Omega$ the generalized
  potential $\Phi(x)$ is essentially a parabola
  with principal axes $K_i>0$  and the diameter of $\omega$ is smaller
  than $(\sum_{i=1}^d K_i)^{-1/2}$, it holds that
\begin{equation*}
  \int_\omega \lvert F_{\rm irr}\cdot\nabla\Phi
    \rvert^2 dx \lesssim \left(\sum_{i=1}^d K_i\right)\, {\rm diam}(\omega)^2 \int_\omega \lvert \nabla\times F_{\rm irr} \rvert^2 dx,
\end{equation*}  
where in $d=3$ the above result is subject to some symmetry
requirements. The above result suggests that in the
vicinity of
local minima of $\Phi(x)$ level-set variations of $\Ji(x)$ are
sub-dominant and we may essentially regard $\Ji(x)$ as being
orthogonal to $\nabla\Phi(x)$ on every subdomain $\omega$. This seems to be particularly important in the $\|D\|\to 0$
limit. 
\end{enumerate}

The fact that the upper bound on $\dot{\Sigma}_s$ depends on the
vorticity of $F_{\rm irr}$ is expected, since vortices nominally break the
forward-backward symmetry of sample paths between two given
points. However, the remaining two contributions seem less
obvious. That is, $\dot{\Sigma}_s$ essentially reflects the energetic cost of
maintaining the nonequilibrium steady state and thus it \emph{a priori} does not seem
to be obvious that $F_{\rm irr}$ with a component in the
direction of $\nabla\Phi(x)$ should ``cost'' more.

The seemingly most insightful result is that in the Freidlin-Wentzel limit
($\|D\|\to 0$), which physically would correspond to a $\Phi(x)$ with
deep minima separated by high barriers. Here non-vanishing $F_{\rm
  irr}(x)$ imply that transitions between two minima in opposite
directions may occur along distinct \emph{instanton} pathways (see
e.g.\ \cite{Stein,Bouchet}).  These instantons are minimum-action
paths between basins around which the reactive sample paths $x_\tau$
concentrate tightly as $\|D\|\to 0$ \cite{Stein,Bouchet}.
On long time-scales, longer
than the typical escape time over a typical barrier of $\Phi(x)$ which
is exponential in the barrier height \cite{Freid,Stein,Bouchet}, such
dynamics has an effective representation with a Markov jump process
on a state space $\{i\}$ \cite{Moro} which is typically assumed to 
obey ``local detailed balance'' \cite{Falasco,Hartich}, i.e.\ with transition rates
between states $i$ and $j$ given by
\begin{equation}
  k_{i\to j}/k_{j\to i}={\rm e}^{\phi(i)-\phi(j)-A_{ij}}
  \label{LDB}
\end{equation}
with antysmmetric link ``affinity'' $A_{ij}=-A_{ji}$. Note that
here $\phi_i$ include, in addition to $\Phi(x)$,
the entropy of the basin \cite{Falasco,Hartich}. 
On such a reduced state space $\dot{\Sigma}_s$ and hence
time-irreversibility is determined solely by the affinities
$A_{ji}\ne 0$ \cite{Seifert2012RPP}. In full generality the
problem of determining $k_{i\to j}$ for an irreversible diffusive dynamics
\eqref{SDE} in the $\|D\|\to 0 $ limit remains unsolved. Important
progress has been made under the assumption of a \emph{transverse
decomposition} \cite{Bouchet} which in our setting corresponds to
the simplifying assumption that $A$ in Eq.~\eqref{decomp} is a
constant matrix and hence \blue{$F_{\rm irr}(x)\perp F_{\rm rev}(x)$}
everywhere. Under these assumptions, the main result of \cite{Bouchet}
(see Eq.~(1.10) in conjunction with Eq.~(2.12) therein) indeed implies
the structure of Eq.~\eqref{LDB}.\\
\indent The asymptotic
correspondence seems to be reflected in our results; i.e.\
according to Corollary~\ref{Freidlin-Wentzel} in the
limit $\|D\|\to 0$ 
the affinities $A_{ij}$ in Eq.~\eqref{LDB}
should indeed arise from local contributions to
$\nabla\Phi^\epsilon(x)\times F^\epsilon_{\rm irr}(x)$ that 
exclude any
contributions parallel to $\nabla\Phi^\epsilon(x)$ along the distinct
forward-backward minimum-action (instanton) transition paths between basins.
This seems to agree with the notion of jump rates obeying local
detailed balance insofar as our upper bounds are reasonably sharp.

\section{Setup}\label{sec3}
Let $\Omega$ be an open subset of $\mathbb{R}^d$ with boundary
$\Gamma$, along which $n$ denotes the unit normal vector. Consider the
overdamped diffusion process in Eq.~\eqref{SDE} on $\Omega$ with
steady state density $\rhoi(x)$, which is the unique normalised
solution of the boundary value problem
$$
\begin{dcases*}
  \nabla\cdot[D\nabla\rhoi(x)-F(x)\rhoi(x)]=0 \quad \text{in}\ \Omega\nonumber\\
[D\nabla\rhoi(x)-F(x)\rhoi(x)]\cdot n =0 \quad \text{on}\ \Gamma.
\end{dcases*}
$$
In the following we define the main quantities of interest. 

\begin{definition}
The invariant current is defined as 
\begin{equation*}
    \Ji(x)\equiv F(x)\rhoi(x)-D\nabla\rhoi(x)=F_{\rm irr}(x)\rhoi(x),
\end{equation*}
the ``non-equilibrium'' potential as
\begin{equation*}
    \Phi(x) \equiv - \ln \rhoi(x),
\end{equation*}
where from follows the reversible part of $F(x)$, defined as 
\begin{equation*}
    F_{\rm rev}(x)\equiv -D\nabla \Phi(x),
\end{equation*}
while the irreversible part of $F(x)$ is defined as 
\begin{equation*}
    F_{\rm irr}(x)\equiv F(x)+D\nabla\Phi(x)=-A(x)\nabla\Phi+\nabla\cdot A(x).
\end{equation*}
The $\rhoi$-weighted vorticity of the drift field $F(x)$ is defined as
\begin{equation*}
    V\equiv \int_{\Omega} \lvert \nabla\times F_{\rm irr}(x) \rvert
    ^2\rhoi(x) dx\overset{D={\rm const}\mathbbm{1}}{=} \int_{\Omega} \lvert \nabla\times F(x) \rvert ^2\rhoi(x) dx,
\end{equation*}
where the last equality holds provided that $D$ is isotropic i.e. a
constant multiple of the identity matrix.
\end{definition}

To demonstrate our results, we need to further introduce some
preliminaries on Sobolev inequalities. We will now state the theorems to
be applied in the derivation of our results and sketch some proofs,
while for the technical details we refer to  \cite{V.Girault} (Chapter
I \S{1}-\S{3}).

We will throughout  assume that $\Omega$ is a bounded,
simply-connected open subset of $\mathbb{R}^d$ with $d=2$ or $3$, and
the boundary $\Gamma=\bigcup_{i=0}^p \Gamma_i$ is
Lipschitz-continuous. Recall that a region is simply connected if every closed curve within it can be continuously deformed to a single point in that region. In $\mathbb{R}^2$, a simply connect region has no holes and the boundary $\Gamma$ consists of only one piece. In $\mathbb{R}^3$, we can allow for holes in the middle without going all the way through the region, and we can have $p\geq 1$. The duality between $H^{-1/2}(\Gamma_i)$ and
$H^{1/2}(\Gamma_i)$ will be denoted by $\langle.,.\rangle_{\Gamma_i}$, and
$(\cdot,\cdot)$ will be used to represent the scalar product on
$L^2(\Omega)^d$. We will present the results for dimension $d=3$, the
same results also hold for dimension $d=2$ for which we refer to
e.g. \cite{V.Girault} (Proposition 3.1) for details.

Let $v:\Omega \to \mathbb{R}^d$ be a vector field. We set up the notations for the following $L^2$-based Sobolev spaces:
\begin{equation*}
    H(div;\Omega)\equiv\{v\in L^2(\Omega)^d:\nabla\cdot v\in L^2(\Omega) \}
\end{equation*}
which is a Hilbert space equipped with the norm
\begin{equation*}
\|v\|_{H(div;\Omega)}=[\|v\|^2_{L^2(\Omega)}+\|\nabla\cdot v\|^2_{L^2(\Omega)}]^{1/2}.
\end{equation*}
We define the subspace of vector fields with zero flux across $\Gamma$: 
\begin{equation*}
 H_0(div;\Omega)\equiv\{u\in H(div;\Omega):u\cdot n\vert _\Gamma =0\},  
\end{equation*}
as well as the subspace with both zero flux across $\Gamma$ and zero divergence in $\Omega$:
\begin{equation*}
    H\equiv\{u\in L^2(\Omega)^d :u\cdot n\vert _\Gamma =0,\nabla\cdot u=0\}.
\end{equation*}  
As for the space $H(div;\Omega)$, we introduce the space $H(curl;\Omega)$:
\begin{equation*}
 H(curl;\Omega)\equiv\{v\in L^2(\Omega)^d;\nabla\times v\in L^2(\Omega)^d \},   
\end{equation*}
which is a Hilbert space equipped with the norm
\begin{equation*}
 \|v\|_{H(curl;\Omega)}=[\|v\|^2_{L^2(\Omega)}+\|\nabla\times v\|^2_{L^2(\Omega)} ]^{1/2}. 
\end{equation*}
Finally, we define
\begin{eqnarray*}
U &\equiv& H_0(div;\Omega)\cap H(curl;\Omega)\\
&=& \{\phi\in L^2(\Omega)^d:\nabla\cdot \phi\in L^2(\Omega),\nabla\times\phi\in L^2(\Omega)^d, \phi\cdot n\vert _\Gamma=0 \} 
\end{eqnarray*}
and the divergence-free subspace
\begin{eqnarray*}
U_0 &\equiv& \{\phi\in U; \nabla\cdot \phi=0\}\\
&=& \{\phi\in L^2(\Omega)^d:\nabla\cdot \phi=0,\nabla\times\phi\in L^2(\Omega)^d, 
\phi\cdot n\vert_\Gamma=0 \}.   
\end{eqnarray*}
$U$ is a Hilbert space equipped with the norm 
\begin{equation*}
\|\phi\|=\|\phi\|_{L^2(\Omega)}+\|\nabla\cdot\phi\|_{L^2(\Omega)}+\|\nabla\times\phi\|_{L^2(\Omega)}.   
\end{equation*}
We now state some preparatory theorems. 

\begin{theorem}[surjectivity of curl]
A vector field $v\in L^2(\Omega)^3$ satisfies
\begin{equation}\label{3.1}
    \nabla\cdot v=0, \quad \langle v\cdot n,1\rangle_{\Gamma_i}=0 \quad \text{for}\ 0\leq i\leq p
\end{equation}
if and only if there exists a vector potential $\phi$ in $H^1(\Omega)^3$ such that 
\begin{equation*}
    v=\nabla\times \phi.
\end{equation*}
Furthermore,
\begin{equation*}
    \nabla\cdot\phi=0.
\end{equation*}
\end{theorem}

\begin{proof}
We outline a sketch of the proof for the 'only if' part. Let $v\in
L^2(\Omega)^3$ satisfy \eqref{3.1}. The compatibility condition
$\langle v\cdot n,1\rangle_{\Gamma_i}=0, \ 0\leq i\leq p$ ensures that we can extend $v$ to the whole $\mathbb{R}^3$ so that the extended function $\Tilde{v}$ belongs to $L^2(\mathbb{R}^3)^3$, is divergence-free, and has compact support. Then the Fourier transform of $\Tilde{v}$, denoted by $\mathcal{F}\Tilde{v}$, is holomorphic, since $\Tilde{v}$ has compact support. After taking the Fourier transform, the conditions $\nabla\cdot\Tilde{v}=0$, $\Tilde{v}=\nabla\times \phi$ and $\nabla\cdot\phi=0$ convert to three algebraic equations which unique determine $\mathcal{F}\phi$, and $\phi$ is constructed by the inverse Fourier transform. The fact that $\phi\in H^1(\Omega)^3$ is a consequence of $\mathcal{F}\Tilde{v}$ being holomorphic.
\end{proof}

\begin{theorem}[bijectivity on $U_0$]\label{bijectivity}
Suppose $v\in L^2(\Omega)^3$ satisfies 
\begin{equation*}
    \nabla\cdot v=0, \quad \langle v\cdot n,1\rangle_{\Gamma_i}=0 \quad \text{for}\ 0\leq i\leq p.
\end{equation*}
Then there exists a unique $\phi$ in $H(curl;\Omega)$ such that
\begin{equation*}
    \nabla\times\phi=v,\quad \nabla\cdot\phi=0,\quad \phi\cdot n=0;
\end{equation*}
it is characterized as the unique solution of the boundary value problem:
\begin{equation*}
\begin{dcases}
\phi\in H,\\
-\Delta\phi=\nabla\times v\quad \text{in} \ H^{-1}(\Omega)^3,\\
(\nabla\times \phi-v)\cdot n=0\quad \text{on} \ \Gamma.
\end{dcases}
\end{equation*}
\end{theorem}

\begin{proposition}\label{N}
Let $\Omega$ be a bounded, connected and Lipschitz-continuous open subset of $\mathbb{R}^d$. The inhomogeneous Neumann's problem
$$
(N)\begin{dcases}
    -\Delta u=f\quad \text{in}\ \Omega,\\
    (\nabla u,\nabla v)=-(\Delta u,v)+\langle g,v\rangle_\Gamma\quad \forall v\in H^1(\Omega)
\end{dcases}
$$
with $f\in L^2(\Omega)$ and $g\in H^{1/2}(\Gamma)$ admits a unique solution $\overline{u}\in  H^1(\Omega)/\mathbb{R}$. Moreover, for each function $u$ in the equivalence class $\overline{u}$, we have 
\begin{equation*}
    \lvert u\rvert_{H^1(\Omega)}\leq C(\|f\|_{L^2(\Omega)}+\|g\|_{H^{-1/2}(\Gamma)}),
\end{equation*}
where $\lvert \cdot \rvert_{H^1(\Omega)}$ denotes the seminorm on $H^1(\Omega)$, and the constant $C$ only depends on $\Omega$.
\end{proposition}

\begin{proof}
One can prove that $\overline{v}\to \lvert v\rvert_{H^1(\Omega)}$ defines a norm equivalent to that of the quotient norm of the Hilbert space $H^1(\Omega)/\mathbb{R}$. Then the problem $(N)$ can be equivalently phrased as
\begin{equation*}
    a(\overline{u},\overline{v})=l(\overline{v})\quad \forall \overline{v}\in H^1(\Omega)/\mathbb{R},
\end{equation*}
with the continuous elliptic bilinear form
\begin{equation*}
    a(\overline{u},\overline{v})\equiv (\nabla u,\nabla v) \quad u\in \overline{u},\ v \in \overline{v}
\end{equation*}
and the linear functional
\begin{equation*}
    l(\overline{v})\equiv (f,v)+\langle g,v\rangle_{\Gamma} \quad v\in \overline{v}.
\end{equation*}
The existence of a unique solution for $u\in H^1(\Omega)/\mathbb{R}$ follows from the Lax-Milgram theorem, and the inequality follows from the operator norm $\|l\|\leq \|f\|_{L^2(\Omega)}+\|g\|_{H^{-1/2}(\Gamma)}$ and the ellipticity of $a(\cdot,\cdot)$.
\end{proof}

\begin{theorem}\label{theoremU}
Suppose $\Omega$ is bounded, Lipschitz-continuous and simply-connected. Then the mapping $\phi\to \nabla\times \phi$ is an isomorphism from the space $U_0$ onto the space:
\begin{equation*}
    T=\{u\in L^2(\Omega)^3; \nabla\cdot u=0,\quad \langle u\cdot
    n,1\rangle_{\Gamma_i}=0,\quad 0\leq i\leq p\};
\end{equation*}
and there exist two positive constants $C_1$ and $C_2$ such that 
\begin{equation}\label{ineq2.1}
    \|\phi\|_{L^2(\Omega)}\leq C_1\|\nabla\times \phi\|_{L^2(\Omega)} \quad \forall \phi \in U_0
\end{equation}
\begin{equation}\label{ineq2.2}
    \|\phi\|_{L^2(\Omega)}\leq C_2\{\|\nabla\times \phi\|_{L^2(\Omega)}+\|\nabla\cdot \phi\|_{L^2(\Omega)}\} \quad \forall \phi \in U
\end{equation}
\end{theorem}

\begin{proof}
The key ingredient in establishing the inequality \eqref{ineq2.1} is
the inverse mapping theorem: the curl map $\nabla\times :U_0\to T$ is
a continuous bijective linear operator from the Banach space $U_0$
onto the Banach space $T$; by the inverse mapping theorem it has a
bounded inverse $(\nabla\times)^{-1}: T\to U_0$. The bijectivity of
the curl map is established in Theorem \ref{bijectivity}, while the continuity automatically follows from the definition of the subspace norm on $U_0$ inherited from $U$.

Every vector field $\phi\in L^2(\Omega)^3$ has the decomposition $\phi=\nabla q+\nabla\times \psi$, where $q\in H^1(\Omega)/\mathbb{R}$ is the only solution of
\begin{equation*}
    (\nabla q,\nabla \mu)=(\phi,\nabla \mu) \quad \forall \mu \in H^1(\Omega).
\end{equation*}
If $\phi \in U$, then $\nabla \psi\in U_0$, and applying Proposition
\ref{N} to $\nabla q$ yields $\|\nabla q\|_{L^2(\Omega)}= \lvert
q\rvert _{H^1(\Omega)}\leq C\|\Delta q\|_{L^2(\Omega)}=C\|\nabla\cdot
\phi\|_{L^2(\Omega)}$. Then \eqref{ineq2.2} follows from
\eqref{ineq2.1} and the triangle inequality.
\end{proof}

We are now in the position to prove our main result.

\section{Proof of the main results}\label{sec4}
 

Since $D$ is a $d\times d$ symmetric
positive definite  matrix, the eigenvalues of $D$ are strictly positive real numbers. Then
\begin{equation*}
    \dot{\Sigma}_s=\int_{\Omega} F_{\rm irr}(x)\cdot D^{-1}F_{\rm irr}(x)\rhoi(x) dx\leq \lambda^{-1}\int_{\Omega}\lvert F_{\rm irr}(x)\rvert^2 \rhoi(x) dx
\end{equation*}
where $\lambda$ denotes the smallest eigenvalue of $D$. For convenience we shall assume $\lambda=1$, hence to control $\dot{\Sigma}_s$ from above it suffices to find upper bounds for the quantity 
\begin{equation*}
  \int_{\Omega}\lvert F_{\rm irr}(x)\rvert^2 \rhoi(x) dx=\int_{\Omega}\lvert \Ji(x)\rvert^2 \rhoi^{-1}(x) dx 
\end{equation*}
subject to the conditions $\nabla\cdot (\Ji)=0$ in $\Omega$ and  $\Ji \cdot n=0$ on $\Gamma$.

\begin{proposition}
Suppose the domain $\Omega$ is bounded, Lipschitz-continuous and simply-connected. Then there exists a positive constant $C(\Omega)$, depending only on $\Omega$, such that
\begin{equation*}
    \int_{\Omega}\lvert \Ji\rvert^2 dx\leq C(\Omega)\int_{\Omega}\lvert \nabla\times\Ji\rvert^2 dx
\end{equation*}
holds for any steady-state current $\Ji$ in $\Omega$.
\end{proposition}

\begin{proof}
This follows directly from the inequality \eqref{ineq2.1}.
\end{proof}

\begin{theorem}\label{theorem1}
Suppose the domain $\Omega$ is bounded, Lipschitz-continuous and simply-connected. Then there exists a positive constant $C(\Omega)$, depending only on $\Omega$, such that 
\begin{equation*}
    \dot{\Sigma}_s\leq C(\Omega) \bigg(V+\frac{1}{4}\int_{\Omega} \Big( \lvert\nabla\Phi\times F_{\rm irr}\rvert^2+\lvert\nabla\Phi\cdot F_{\rm irr}\rvert^2 \Big)\rhoi \ dx  \bigg)
\end{equation*}
holds at the non-equilibrium steady state of any diffusion process in $\Omega$.
\end{theorem}

\begin{proof}
Since $\Ji\cdot n= (F_{\rm irr}\rhoi)\cdot n=0$ on $\Gamma$, it follows that $(F_{\rm irr}\rhoi^{1/2})\cdot n=0$ on $\Gamma$. Applying Theorem \ref{theoremU} gives
\begin{align*}
\int_\Omega \lvert F_{\rm irr}\rhoi^{1/2}\rvert^2 dx
   &\leq C_1  \bigg(\int_\Omega \lvert\nabla\times (F_{\rm irr}\rhoi^{1/2})\rvert^2 dx+\int_\Omega \lvert\nabla\cdot (F_{\rm irr}\rhoi^{1/2})\rvert^2 dx  \bigg)\\
   &= C_1 \bigg(\int_\Omega \lvert\rhoi^{1/2}\nabla\times F_{\rm irr}+\nabla\rhoi^{1/2}\times F_{\rm irr} \rvert^2 dx \\
   &\quad +\int_\Omega \lvert\rhoi^{1/2}\nabla\cdot F_{\rm irr} +\nabla\rhoi^{1/2}\cdot F_{\rm irr} \rvert^2 dx  \bigg)\\
   &\leq C_2 \int_\Omega \Big( \lvert\nabla\times F_{\rm irr}\rvert^2 +\frac{1}{4}\lvert\nabla\Phi\times F_{\rm irr}\rvert^2+\frac{1}{4}\lvert\nabla\Phi\cdot F_{\rm irr}\rvert^2 \Big)\rhoi  dx,
\end{align*}
where to obtain the last inequality we have used 
\begin{equation*}
\nabla\rhoi^{1/2}=\frac{1}{2}\rhoi^{-1/2}\nabla\rhoi=-\frac{1}{2}\rhoi^{1/2}\nabla\Phi
\end{equation*}
and 
\begin{equation*}
\nabla\cdot(F_{\rm irr}\rhoi)=\rhoi\nabla\cdot F_{\rm irr}+\nabla\rhoi\cdot F_{\rm irr}=0 \implies \nabla\cdot F_{\rm irr}=\nabla\Phi\cdot F_{\rm irr}.
\end{equation*}
\end{proof}

\begin{corollary}(the Freidlin-Wentzel limit)\label{Freidlin-Wentzel}
Suppose $\epsilon>0$, consider the steady state of the diffusion process in $\Omega$
\begin{equation*}
dx_\tau=F(x_\tau)d\tau+\sqrt{2\epsilon \blue{\hat{D}}}dW_\tau,
\end{equation*}
where we use superscripts to emphasize the quantities' dependence on
$\epsilon$, and we introduced $\hat{D}\equiv D/\|\blue{D}\|$. Then, as $\epsilon$ approaches zero, at the leading order in $\epsilon$ we have 
\begin{equation*}
\dot{\Sigma}^\epsilon_s\leq C(\Omega) \int_{\Omega}
\lvert\nabla\Phi^\epsilon\times F_{\rm irr}
\rvert^2\rhoi^\epsilon  dx\overset{D={\rm const}\mathbbm{1}}{=} C(\Omega) \int_{\Omega}
\lvert\nabla\Phi^\epsilon\times F
\rvert^2\rhoi^\epsilon \ dx,  
\end{equation*}
where $C(\Omega)$ is a constant depending only on the
domain. 
\end{corollary}

\begin{proof}
According to the Large Deviation Theory \cite{Freid} the limit
\begin{equation*}
    \lim_{\epsilon\to 0}\epsilon \ln \rhoi^\epsilon(x) =:-\psi_0(x)
\end{equation*}
exits under the assumed confining conditions on $F$. Since $\Phi^\epsilon(x)\equiv -\ln \rhoi^\epsilon(x)$, we write out the asymptotic expansion for $\Phi^\epsilon(x)$ as
\begin{equation*}
\Phi^\epsilon(x)=\epsilon^{-1}\psi_0(x)+\psi_1(x)+O(\epsilon).
\end{equation*}
The decomposition \eqref{Fdecomp} further gives the following representation formula for $F$
\begin{equation}
    F(x)=-A^\epsilon(x) \nabla\Phi^\epsilon(x)+\nabla\cdot
    A^\epsilon(x)-\epsilon \blue{\hat{D}}\nabla\Phi^\epsilon(x).
\label{repr}
\end{equation}
Since $\nabla\Phi^\epsilon(x)$ is of order $\epsilon^{-1}$,
Eq.~\eqref{repr} suggests that $\nabla\cdot A^\epsilon(x)$ is of order
$\epsilon$. Therefore, both $\nabla\times F^\epsilon_{\rm irr}(x)$ and
$\nabla\Phi^\epsilon(x)\cdot F^\epsilon_{\rm
  irr}(x)=\nabla\Phi^\epsilon(x)\cdot (\nabla\cdot A^\epsilon(x))$ are of
order $1$, while $\nabla\Phi^\epsilon(x)\times F^\epsilon_{\rm
  irr}(x)$ is of order $\epsilon^{-1}$. Further, we have
\begin{equation*}
\nabla\Phi^\epsilon(x)\times F(x)=\nabla\Phi^\epsilon(x)\times F_{\rm irr}(x)-\epsilon\nabla\Phi^\epsilon(x)\times\hat{D}\nabla\Phi^\epsilon(x)
\end{equation*}
and if $D={\rm const}\cdot\mathbbm{1}$ then also \blue{$\nabla\Phi^\epsilon(x)\times F(x)=\nabla\Phi^\epsilon(x)\times F_{\rm irr}(x)$}. Thus, to leading order in $\epsilon$ only contributions
$\perp\nabla\Phi$ are important.
The conclusion follows from Theorem \ref{theorem1}.
\end{proof}

\begin{theorem}\label{theorem2}
Suppose the domain $\Omega$ is bounded, Lipschitz-continuous and simply-connected, and the steady state density $\rhoi$ satisfies
\begin{equation*}
    \rhoi\vert_\Gamma >0;\quad \sup_\Omega \rhoi<\infty;\quad \rhoi^{1/(1-q)}\in L^1(\Omega),\ 1<q<\infty.
\end{equation*}
Then there exists a constant $C(\Omega,\rhoi)>0$, depending on $\Omega$ and $\rhoi$, such that 
\begin{equation*}
    \dot{\Sigma}_s\leq C(\Omega,\rhoi) \left(
    \left(\int_\Omega \lvert\nabla\times F_{\rm irr}\rvert^{2q}\rhoi dx \right)^{1/q}
    +\left(\int_\Omega \lvert\nabla\Phi\cdot F_{\rm irr}\rvert^{2q}\rhoi dx\right)^{1/q} \right)
\end{equation*}
and the constant $C(\Omega,\rhoi)$ takes the form
\begin{equation*}
    C(\Omega,\rhoi)=C(\Omega)\cdot \sup_\Omega \rhoi\cdot  \left(\int_\Omega \rhoi^{1/(1-q)}dx \right)^{(q-1)/q}.
\end{equation*}
\end{theorem}

\begin{proof}
We have
\begin{equation*}
 \int_{\Omega}\lvert F_{\rm irr}\rvert^2\rhoi dx\leq  \sup_\Omega \rhoi \int_{\Omega}\lvert F_{\rm irr}\rvert^2 dx.       
\end{equation*}
The assumption $\rhoi>0$ on $\Gamma$ implies $F_{\rm irr}\cdot n\vert_\Gamma=0$, so we apply Theorem \ref{theoremU} to $F_{\rm irr}$ to get
\begin{align*}
\int_{\Omega}\lvert F_{\rm irr}\rvert^2 dx
&\leq C(\Omega)\left(\int_{\Omega}\lvert\nabla\times F_{\rm irr}\rvert^2dx+\int_{\Omega}\lvert \nabla\cdot F_{\rm irr}\rvert^2dx \right)\\
&= C(\Omega)\left(  \int_\Omega \lvert \nabla\times F_{\rm irr}\rvert ^2\rhoi^{1/q}\rhoi^{-1/q} dx +\int_\Omega \lvert \nabla\Phi\cdot F_{\rm irr}\rvert ^2\rhoi^{1/q}\rhoi^{-1/q}dx \right)\\
&\leq C(\Omega)\left(\int_\Omega \rhoi^{1/(1-q)}dx \right)^{(q-1)/q}\\
&\quad \left(\left(\int_\Omega \lvert\nabla\times F_{\rm irr}\rvert ^{2q}\rhoi dx\right)^{1/q}+\left(\int_\Omega \lvert \nabla\Phi\cdot F_{\rm irr}\rvert ^{2q}\rhoi dx\right)^{1/q}\right),
\end{align*}
where the constant $C(\Omega)>0$ only depends on $\Omega$, and in the last step we have used H\"older's inequality.    
\end{proof}

\begin{remark}
Since $\rhoi^{1/(1-q)}\in L^1(\Omega)$ if and only if  $\rhoi^{-1}\in L^{1/(q-1)}(\Omega)$, to have $q$ approach $1$ we require higher integrability of $\rhoi^{-1}$. If 
\begin{equation*}
   \rhoi^{-1}\in L^{\infty}(\Omega)\quad \text{i.e.}\ \inf_\Omega\rhoi>0,
\end{equation*}
we can set $q=1$ to get
\begin{equation*}
    \dot{\Sigma}_s\leq C(\Omega,\rhoi) \left(V+\int_\Omega \lvert\nabla\Phi\cdot F_{\rm irr}\rvert ^2\rhoi dx\right),
\end{equation*}
where the constant $C(\Omega,\rhoi)$ takes the form
\begin{equation*}
    C(\Omega,\rhoi)=C(\Omega)\cdot \sup_\Omega \rhoi\cdot \left(\inf_\Omega \rhoi \right)^{-1}, 
\end{equation*}
and in such cases $\rhoi$ defines a measure on $\Omega$ which is equivalent to the Lebesgue measure.
\end{remark}

\begin{corollary}\label{bounds}
Suppose the domain $\Omega$ is bounded, Lipschitz-continuous and
simply-connected, and the steady state density $\rhoi$ is strictly positive on $\Gamma$. Then there exists a constant $C(\Omega)>0$ depending only on $\Omega$, such that 
 \begin{equation*}
    \dot{\Sigma}_s\leq C(\Omega)\cdot \sup_\Omega \rhoi \left(\int_{\Omega}\lvert\nabla\times F_{\rm irr}\rvert^2dx+\int_{\Omega}\lvert \nabla\Phi\cdot F_{\rm irr}\rvert^2dx \right)
\end{equation*}
and 
\begin{equation*}
    \dot{\Sigma}_s\leq C(\Omega)\cdot \sup_\Omega \rhoi\cdot \left(\inf_\Omega \rhoi \right)^{-1} \left(V+\int_{\Omega}\lvert \nabla\Phi\cdot F_{\rm irr}\rvert^2\rhoi dx \right).
\end{equation*}
\end{corollary}

\begin{proof}
The statements follow from the lines of the proof of Theorem \ref{theorem2}. If $\inf_\Omega \rhoi =0$ we interpret the right-hand-side as infinity.
\end{proof}


We have shown in Corollary \ref{Freidlin-Wentzel} that when the steady
state density is peaked, the entropy production rate is dominated by
the vorticity-like quantity $\int_{\Omega}\lvert \nabla\Phi\times
F\rvert ^2\rhoi \ dx$. Conversely, if the potential $\Phi(x)$ is a
parabola such that $\lvert\nabla\Phi(x)\rvert$ is small near the
minimum of $\Phi(x)$ (and thus where the mass of $\rhoi$ is concentrated), e.g.\ $\rhoi$ is a Gaussian centered at the origin and $\nabla\Phi(x)=cx$ for some constant $c$, then the upper-bounds in Corollary \ref{bounds} are also dominated by vorticity terms.

\begin{corollary}
Suppose $\omega\subset\Omega$, $\Phi(x)$ is a parabola on $\omega$ with principal axes $\partial^2_{x_i}\Phi(x)=K_i>0$, $1\leq i\leq d$, ${\rm detHess}(\Phi)\geq 0$ and $\omega$ contains a local minimum of $\Phi$. For $d=2$, if the diameter of $\omega$ is much smaller than $(\Delta\Phi)^{-1/2}$ then $\int_\omega \lvert \nabla\times F_{\rm irr} \rvert^2 dx$ dominates over $\int_\omega \lvert F_{\rm irr}\cdot\nabla\Phi \rvert^2 dx$ in the sense that
\begin{equation*}
  \int_\omega \lvert F_{\rm irr}\cdot\nabla\Phi
    \rvert^2 dx \lesssim (\Delta \Phi)\, {\rm diam}(\omega)^2 \int_\omega \lvert \nabla\times F_{\rm irr} \rvert^2 dx.
\end{equation*}
For $d=3$, we have the same conclusion provided that the system possesses certain symmetry, which will be specified in the proof.
\end{corollary}

\begin{proof}
We first deal with the 2-dimensional case as the notation is simpler. We use $(x,y)$ to denote the two independent variables and use the decomposition \eqref{Fdecomp} to write
\begin{equation*}
    F_{\rm irr}=(-a\partial_y \Phi,a\partial_x \Phi)^\intercal+(\partial_y a,-\partial_x a)^\intercal,
\end{equation*}
where $A=
\big(\begin{smallmatrix}
0 & a(x,y)\\
-a(x,y) & 0
\end{smallmatrix}\big)$ and $F_{\rm irr}=-A\nabla\Phi+\nabla\cdot A$. Then
\begin{equation*}
    F_{\rm irr}\cdot\nabla\Phi=-\{\partial_y\Phi\}\{\partial_x a\}+\{\partial_x\Phi\}\{\partial_y a\}
\end{equation*}
where $\{\cdot\}$ signifies that the differentials only act within the bracket and
\begin{equation*}
    \nabla\times F_{\rm irr}=\{\partial_y a\}\{\partial_y\Phi\}+\{\partial_x a\}\{\partial_x\Phi\}-\Delta a+a\Delta \Phi.
\end{equation*}
Since $\Delta \Phi=\sum_{i=1}^dK_i\equiv K$ is a positive constant, we can prove that
$\|(-\Delta+K)a(x,y)\|_{L^2(\omega)}$ bounds $\lvert a\rvert
_{H^1(\omega)}$ from above using the Fourier transform $\mathcal{F}$:
\begin{equation*}
\lvert\mathcal{F}(-\Delta+K)a\rvert=(K+\xi_1^2+\xi_2^2)\lvert \mathcal{F}a \rvert \geq \sqrt{K}(\lvert\xi_1\rvert+\lvert\xi_2\rvert)\lvert \mathcal{F}a \rvert = \sqrt{K}(\lvert\mathcal{F}\partial_x a \rvert+ \lvert\mathcal{F}\partial_y a \rvert)
\end{equation*}
together with the Parseval's identity, which asserts that the Fourier
transform on $L^2(\mathbb{R}^d)$ preserves the $L^2$ norm. If $\omega$
contains a local minimum of $\Phi$, say at $(x_0,y_0)\in \omega$, then $\nabla\Phi(x_0,y_0)=0$ and
\begin{align*}
    \lvert \partial_x\Phi(x,y)\rvert =&\lvert \int_{0}^{1}\partial^2_x \Phi(x_0+s(x-x_0),y_0+s(y-y_0))(x-x_0)\\
    &+\partial^2_{xy} \Phi(x_0+s(x-x_0),y_0+s(y-y_0))(y-y_0) ds\rvert \leq K\text{diam}(\omega).
\end{align*}
where the mixed derivative of $\Phi$ is controlled by $K$ because ${\rm detHess}(\Phi)\geq 0$. The same uniform bound also holds for $\lvert \partial_y\Phi(x,y)\rvert$. Putting everything together and setting $\lvert \Phi\rvert_1\equiv \sup_\omega\{\lvert\partial_x\Phi(x,y)\rvert,\lvert\partial_y\Phi(x,y)\rvert\}$, we have
\begin{align*}
    \int_\omega \lvert F_{\rm irr}\cdot\nabla\Phi
    \rvert^2 dx &\lesssim \lvert \Phi\rvert_1^2\lvert a\rvert^2_{H^1(\omega)}\\
    &\lesssim K^{-1}\lvert \Phi \rvert^2_1\|(-\Delta+K)a(x,y)\|^2_{L^2(\omega)}\\
    &\lesssim K^{-1}\lvert \Phi \rvert^2_1\int_\omega \lvert \nabla\times F_{\rm irr} \rvert^2 dx \\
    &\lesssim K\, \text{diam}(\omega)^2 \int_\omega \lvert \nabla\times F_{\rm irr} \rvert^2 dx.
\end{align*}
Therefore, if the diameter of the subdomain $\omega$ is much smaller than $1/\sqrt{K}$, the vorticity is dominating over the level-set variations.

Now we consider the 3-dimensional case and use $(x,y,z)\in \mathbb{R}^3$ to denote the independent variables. Let 
\begin{equation*}
A=\begin{pmatrix}
0 & a_1(x,y,z) &  a_2(x,y,z) \\
-a_1(x,y,z) & 0 & a_3(x,y,z) \\
-a_2(x,y,z) & -a_3(x,y,z) & 0
\end{pmatrix}
\end{equation*}
and set $F_{\rm irr}=-A\nabla\Phi+\nabla\cdot A$. We compute
\begin{align*}
    F_{\rm irr}\cdot\nabla\Phi&=\{\partial_z a_3\}\{\partial_y\Phi\}-\{\partial_x a_1\}\{\partial_y\Phi\}-\{\partial_y a_3\}\{\partial_z\Phi\}-\{\partial_x a_2\}\{\partial_z\Phi\}\\&+\{\partial_z a_2\}\{\partial_x\Phi\}+\{\partial_y a_1\}\{\partial_x\Phi\},
\end{align*}
which involves only first order derivatives of $a_i$ and
$\Phi$. Moreover, we have
\begin{align*}
    (\nabla\times F_{\rm irr})_1 =&-\partial^2_z a_3+a_3\partial^2_z\Phi-\partial^2_y a_3+a_3\partial^2_y\Phi\\
    &+\partial^2_{xz} a_1-a_1\partial^2_{xz}\Phi-\partial^2_{xy} a_2+a_2\partial^2_{xy}\Phi\\
    &+\{\partial_z a_3\}\{\partial_z\Phi\}+\{\partial_y a_3\}\{\partial_y\Phi\}-\{\partial_z a_1\}\{\partial_x\Phi\}+\{\partial_y a_2\}\{\partial_x\Phi\},
\end{align*}
and similar expressions for the other two components. Now the
difficulty arises because the second-order differential operators on
$a_i's$ are no longer elliptic, and we require certain
symmetry conditions to control the mixed derivatives as well as the first order terms. For instance, the assertion of the theorem holds if the three spacial dimensions are independent and self-similar, meaning that $\partial^2_{x_i x_j}\Phi=0,\ \forall i\neq j$, $\partial^2_{x_ix_j}a_k=0,\, \forall i\neq j,\,1\le k\le 3$ and $\{\partial_{x_i}a_j: 1\leq i,j\leq 3\}$ are multiples of one another. Then $\|\nabla \times F_{\rm irr}\|_{L^2(\omega)}$ bounds $\|F_{\rm
  irr}\cdot\nabla\Phi\|_{L^2(\omega)}$ from above, and it is dominating when the diameter of $\omega$ is small relative to $1/\sqrt{K}$.
\end{proof}

\bmhead{Acknowledgments}

Financial support from the German Research Foundation (DFG) through the Emmy Noether Program GO 2762/1-2 (to A.~G.) is gratefully acknowledged.

\bmhead{Data availability} Data sharing not applicable to this article
as no datasets were generated or analysed during the current study.

\bmhead{Declarations} The authors have no competing interests to declare that are relevant to the content of this article.

\bibliography{vorticity.bib}


\end{document}